\documentclass[journal,twoside,web,letter]{ieeecolor_noruler}
\usepackage{generic}
\usepackage{cite}
\usepackage{amsmath,amssymb,amsfonts}
\usepackage{algorithmic}
\usepackage{graphicx}
\usepackage{algorithm,algorithmic}
\usepackage{amssymb}
\usepackage{hyperref}
\hypersetup{hidelinks=true}
\usepackage{textcomp}
\usepackage{wrapfig}
\usepackage{nicematrix}
\usepackage{mathrsfs}

\usepackage{tikz}
\usepackage{tikz-cd} 
\usepackage{caption}  
\usetikzlibrary{matrix}
\usepackage{graphicx}
\usetikzlibrary{matrix,arrows,decorations.pathmorphing}
\usepackage{arydshln}
\usepackage[columnwise,mathlines]{lineno}
\newtheorem{corol}{Corollary}
\newtheorem{definition}{Definition}

\newtheorem{lemma}{Lemma}
\newtheorem{remark}{Remark}
\newtheorem{thm}{Theorem}

\newcommand{\im}{\operatorname{im}}

\DeclareMathOperator{\rank}{rank}
\DeclareMathOperator{\R}{\mathbb{R}}

\def\BibTeX{{\rm B\kern-.05em{\sc i\kern-.025em b}\kern-.08em
		T\kern-.1667em\lower.7ex\hbox{E}\kern-.125emX}}
\begin{document}
	\title{Unified Eigenvalue–Eigenspace Criteria for Functional Properties of Linear Systems  and a Generalized Separation Principle}
	\author{Tyrone Fernando
		\thanks{T. Fernando is with the Department of Electrical, Electronic and Computer Engineering, University of Western Australia (UWA), 35 Stirling Highway, Crawley, WA 6009, Australia. (email:  tyrone.fernando@uwa.edu.au)}}
	
	\maketitle
	
	\begin{abstract}
Classical controllability and observability admit equivalent Popov–Belevitch–Hautus (PBH) tests based on eigenvalue-wise rank conditions. This paper extends this framework to functional properties of linear systems, establishing necessary and sufficient PBH-style conditions for Functional Controllability, Functional Stabilizability, Functional Observability, and Functional Detectability in terms of generalized eigenspaces. By contrast, Target Output Controllability is shown not to admit an independent eigenvalue-wise characterization in general; a necessary and sufficient eigenstructure condition is derived instead. We further introduce Intrinsic Functional Controllability and Intrinsic Functional Stabilizability, which give necessary and sufficient conditions for the existence of the augmentation matrices required for functional controller and observer synthesis. These intrinsic properties yield a Generalized Separation Principle that recovers the classical separation principle as a special case.
	\end{abstract}


\section{Introduction}

Classical controllability and observability are fundamental concepts in linear systems theory, admitting both geometric and Popov--Belevitch--Hautus (PBH) characterizations. Motivated by applications in large-scale and networked systems, increasing attention has been devoted to \emph{functional} generalizations of these properties \cite{ref8n}--\cite{ref9n}. Rather than controlling or estimating the full state vector, the objective is to control or reconstruct a prescribed linear functional of the state vector, representing quantities of practical interest such as aggregate power, population totals, or system performance measures. This shift leads naturally to functional analogues of controllability, stabilizability, observability, and detectability.

Existing work formulates functional system properties within Kalman's geometric framework \cite{ref2ty}. Subspace-based characterizations have been established for Functional Observability \cite{ref10,ref3n,ref10n}, Functional Detectability \cite{ref11n}, Functional Controllability and Functional Stabilizability \cite{ref4n}, and Target Output Controllability \cite{ref5n}. These formulations preserve the classical dualities and reduce to the standard notions when the prescribed functional coincides with the full state. Functional Observability has also been extended to sampled-data, networked, and nonlinear systems \cite{ref12n,ref13n,ref13nc,ref7n,ref16n,ref5ty,ref13nb,15new}.

However, PBH-style characterizations \cite{hautus1969} of functional properties remain limited and have generally been confined to special classes of systems \cite{ref8n,ref7n,ref6n}, often relying on restrictive assumptions such as diagonalizability or particular Jordan structures. Our earlier work \cite{ref4n} introduced functional controllability, stabilizability, observability, and detectability from a geometric (subspace) perspective. The present paper develops the eigenvalue--eigenspace counterpart of that theory, giving necessary and sufficient PBH-style tests for each property.

Furthermore, \cite{ref4n} established a generalized separation principle based on an augmentation-based constructive procedure, for which the existence of the required augmentation matrices cannot, in general, be verified \emph{a priori}. The present paper complements and extends those results by introducing the notions of \emph{Intrinsic Functional Controllability} and \emph{Intrinsic Functional Stabilizability}, which provide verifiable conditions for the existence of the required augmentation matrices and hence for the generalized separation principle.

\noindent\textbf{Contributions.}
Building upon the geometric framework developed in \cite{ref4n}, the main contributions of this paper are as follows:
(i) A unified PBH-style framework is developed, yielding necessary and sufficient eigenvalue-based criteria for Functional Controllability, Functional Stabilizability, Functional Observability, and Functional Detectability.
(ii) A necessary and sufficient eigenstructure characterization of Target Output Controllability is established, showing that this property does not, in general, admit an independent eigenvalue-wise characterization.
(iii) Intrinsic Functional Controllability and Intrinsic Functional Stabilizability are introduced, yielding verifiable conditions for the existence of the augmentation matrices required for functional controller and observer synthesis and establishing a Generalized Separation Principle with the classical separation principle as a special case.

\noindent\textbf{Notation.}
The spectrum of a matrix $A$, i.e., the set of its eigenvalues, is denoted
by $\sigma(A)$. For a complex number $\lambda$, $\Re(\lambda)$ denotes its
real part. 
For a matrix $A$, a nonzero vector $v$ is a \emph{generalized left eigenvector} of order $k$ associated with $\lambda\in\sigma(A)$ if $(\lambda I-A^{\mathsf T})^kv=\mathbf 0$ and $(\lambda I-A^{\mathsf T})^{k-1}v\neq\mathbf 0$; a \emph{generalized right eigenvector} of order $k$ is defined analogously via $(\lambda I-A)^kv=\mathbf 0$ and $(\lambda I-A)^{k-1}v\neq\mathbf 0$. An eigenvector corresponds to the case $k=1$.
The controllability matrix of the pair $(A,B)$ is denoted
by $\mathcal C_{(A,B)}$, and the observability matrix of the pair
$(A,C)$ is denoted by $\mathcal O_{(A,C)}$. For a matrix $M$, $M^{-}$ denotes a generalized inverse satisfying
$MM^{-}M=M$. The identity matrix of appropriate dimensions is denoted by
$I$. The image, kernel, and rank of $M$ are denoted by
$\im(M)$, $\ker(M)$, and $\rank(M)$, respectively. The symbol $\oplus$ denotes the direct sum of subspaces. Throughout the paper, $e_i$ denotes the $i$-th standard
basis vector. A matrix $F$ is said to annihilate a subspace
$\mathcal V$ if $\mathcal V\subseteq\ker(F)$. Throughout this paper, the reachable and unobservable subspaces refer to
$\im\mathcal{C}_{(A,B)}$ and
$\ker\mathcal{O}_{(A,C)}$, respectively. Observable and uncontrollable
subspaces refer specifically to their orthogonal complements,
$\im\mathcal{O}_{(A,C)}^{\mathsf T}$ and
$\ker\mathcal{C}_{(A,B)}^{\mathsf T}$, respectively.

\section{Definitions and Preliminary Results}

Consider the continuous-time linear time-invariant system
\begin{IEEEeqnarray}{rcl}
	\dot{x}(t) &=& Ax(t)+Bu(t) \nonumber\\
	y(t) &=& Cx(t) \nonumber \label{1b}
\end{IEEEeqnarray}
where $x(t)\in\mathbb{R}^n$ is the state vector, $u(t)\in\mathbb{R}^m$ is the input, and $y(t)\in\mathbb{R}^p$ is the measured output. The matrices
$A\in\mathbb{R}^{n\times n}$, $B\in\mathbb{R}^{n\times m}$, and the full row-rank matrix
$C\in\mathbb{R}^{p\times n}$ denote the system, input, and output matrices, respectively.

We begin by recalling the functional properties considered in this paper. These extend the corresponding classical notions by replacing the objective of controlling or estimating the full state vector with the linear functional
\[
z(t)=Fx(t),
\]
where $F\in\mathbb{R}^{r\times n}$. Functional Controllability and Functional Stabilizability are characterized in terms of the controllable and stabilizable subspaces, respectively, while Functional Observability and Functional Detectability are defined dually in terms of the observable and detectable subspaces.

\begin{definition}[Functional Controllability (FC)]\label{def1}
	The functional \(z(t)=Fx(t)\) is controllable, or equivalently the triple
	\((A,B,F^{\mathsf T})\) is \emph{functional controllable} if and only if
	the functional subspace \(\im(F^{\mathsf T})\) is contained in
	the reachable subspace of the pair \((A,B)\), that is,
	\(\im(F^{\mathsf T})\subseteq\im\mathcal{C}_{(A,B)}\).
\end{definition}

The following standard linear-algebraic result is used repeatedly
throughout the paper.
\medskip

\begin{lemma}[Image--Kernel Duality]\label{lem:im-ker-duality}
	Let $X\in\mathbb R^{n\times m_X}$ and $Y\in\mathbb R^{n\times m_Y}$ be matrices with the same number of rows.
	The following statements are equivalent:
	\begin{enumerate}
		\item[i.] \label{it:imker1}$\im(X)\subseteq \im(Y)$
		\item[ii.]\label{it:imker2}$\ker(Y^{\mathsf T})\subseteq \ker(X^{\mathsf T})$
		\item[iii.]\label{it:imker3} $\rank\!\begin{pmatrix} Y & X \end{pmatrix}=\rank(Y)$
		\item[iv.]\label{it:imker4}$\rank\!\begin{pmatrix} Y^{\mathsf T} \\ X^{\mathsf T} \end{pmatrix}=\rank(Y^{\mathsf T})$
	\end{enumerate}
\end{lemma}

Combining Definition \ref{def1} with Lemma \ref{lem:im-ker-duality} yields the following equivalent formulations of Functional Controllability (FC).

\begin{lemma}[Equivalent FC Formulations]\label{lemma2}
	The following statements are equivalent.
	
	\begin{enumerate}
		\item[(i)] The triple
		\((A,B,F^{\mathsf T})\) is \emph{functional controllable}.		
		\item[(ii)]
$		\im(F^{\mathsf T})
		\subseteq
		\im\mathcal C_{(A,B)}.
$		
		\item[(iii)]
$		\ker\!\left(\mathcal C_{(A,B)}^{\mathsf T}\right)
		\subseteq
		\ker(F).
$		
		\item[(iv)]
$		\ker\!\left(\mathcal C_{(A,B)}^{\mathsf T}\right)
		\subseteq
		\ker\!\left(\mathcal C_{(A,F^{\mathsf T})}^{\mathsf T}\right).
$		
		\item[(v)]
$		\rank
		\begin{pmatrix}
			\mathcal C_{(A,B)} & F^{\mathsf T}
		\end{pmatrix}
		=
		\rank\mathcal C_{(A,B)}.
$		
\item[(vi)]
		For every generalized left eigenvector \(v\) of \(A\), if \(v\) belongs to the
		uncontrollable subspace of the pair \((A,B)\), i.e.,
		$
		v\in
		\ker\!\left(\mathcal C_{(A,B)}^{\mathsf T}\right),
		$
		then \(v\) belongs to the uncontrollable subspace of the pair
		\((A,F^{\mathsf T})\), i.e.,
		$
		v\in
		\ker\!\left(\mathcal C_{(A,F^{\mathsf T})}^{\mathsf T}\right).
		$
	\end{enumerate}
\end{lemma}

\begin{proof}
	By Definition~\ref{def1}, statements~(i) and~(ii) are equivalent. Applying
	Lemma~\ref{lem:im-ker-duality} with
$
	X=F^{\mathsf T},$ and
$	Y=\mathcal C_{(A,B)},
$
	establishes the equivalence of statements~(ii), (iii), and~(v).
	Moreover, statements~(iii) and~(iv) are equivalent by Theorem~2 of
	\cite{ref4n}.
	It remains to prove that statements~(iv) and~(vi) are equivalent. Define
	\[
	\mathcal N_B
	:=
	\ker\!\left(\mathcal C_{(A,B)}^{\mathsf T}\right),
	\quad
	\mathcal N_F
	:=
	\ker\!\left(\mathcal C_{(A,F^{\mathsf T})}^{\mathsf T}\right).
	\]
	If statement~(iv) holds, then $\mathcal N_B\subseteq\mathcal N_F$, and
	statement~(vi) follows immediately.
	
	Conversely, suppose statement~(vi) holds. Since $\mathcal N_B$ is
	$A^{\mathsf T}$-invariant, it admits a basis consisting of generalized
	left eigenvectors of $A$. Hence every \(x\in\mathcal N_B\) can be expressed as
	$
	x=\sum_j \alpha_j v_j,
	$
	where each \(v_j\in\mathcal N_B\) is a generalized left eigenvector of \(A\).
	By statement~(vi), $v_j\in\mathcal N_F$ for all $j$. Since $\mathcal N_F$
	is a subspace, $x\in\mathcal N_F$. Therefore,
	$
	\mathcal N_B\subseteq\mathcal N_F,
	$
	which is statement~(iv).
	Thus the six statements are equivalent.
\end{proof}

\begin{remark}
	The term $F^{\mathsf T}$ is used instead of $F$ 
	because Functional Controllability compares the invariant subspaces generated by the pairs $(A,B)$ and $(A,F^{\mathsf T})$.
	Thus, the pair
	$(A,F^{\mathsf T})$ plays the same role for the functional as
	$(A,B)$ does for the control input.
\end{remark}

For comparison, we recall Target Output Controllability.	
\begin{definition}[Target Output Controllability \cite{ref5n, ref15n} ]\label{def2}
	The functional $z(t)=Fx(t)$ is target output controllable, or equivalently, the triple $(A,B,F)$ is target output controllable, if for any initial value $z(t_0)$ and any final target value $z(t_1)$, there exists an input $u(t)$ that steers $z(t_0) = F x(t_0)$ to $z(t_1) = F x(t_1)$ in finite time $t_1 > t_0$.
\end{definition}	

\begin{remark}
	Target Output Controllability is a behavioural input--output reachability property characterized by the well-known rank condition \cite{ref5n},
	$
		\rank\!\bigl(F\,\mathcal C_{(A,B)}\bigr)=\rank(F).
$
	This criterion differs from the Functional Controllability criterion in item~(v) of Lemma~\ref{lemma2}. As established in Theorem~4 of \cite{ref4n}, Functional Controllability implies Target Output Controllability, whereas the converse does not hold. Thus, Functional Controllability is a stronger structural property that guarantees the behavioural reachability property of Target Output Controllability by requiring
$
	\im(F^{\mathsf T})\subseteq\im\mathcal C_{(A,B)}.
$\end{remark}

\begin{definition}[Functional Observability (FO)]\label{def3}
	The functional \(z(t)=Fx(t)\) is observable, or equivalently the triple $(A,C,F)$ is \emph{functional observable}, if and only if the unobservable subspace of the pair
	\((A,C)\) is contained in the kernel of \(F\).
\end{definition}

\begin{lemma}[Equivalent FO Formulations]\label{lemma3}
	The following statements are equivalent.
	
	\begin{enumerate}
		\item[(i)] The triple
		\((A,C,F)\) is \emph{functional observable}.		
		
		\item[(ii)]
		$
		\ker\!\left(\mathcal O_{(A,C)}\right)
		\subseteq
		\ker(F).
		$
		
		\item[(iii)]
		$
		\im(F^{\mathsf T})
		\subseteq
		\im\!\left(\mathcal O_{(A,C)}^{\mathsf T}\right).
	$
		
		\item[(iv)]
	$
		\ker\!\left(\mathcal O_{(A,C)}\right)
		\subseteq
		\ker\!\left(\mathcal O_{(A,F)}\right).
	$
		
		\item[(v)]
$
		\rank
		\begin{pmatrix}
			\mathcal O_{(A,C)}\\
			F
		\end{pmatrix}
		=
		\rank\!\left(\mathcal O_{(A,C)}\right).
	$
		
		\item[(vi)]
		For every generalized right eigenvector \(v\) of \(A\), if \(v\) belongs to the
		unobservable subspace of the pair \((A,C)\), i.e.,
		$
		v\in
		\ker\!\left(\mathcal O_{(A,C)}\right),
		$
		then \(v\) belongs to the unobservable subspace of the pair
		\((A,F)\), i.e.,
		$
		v\in
		\ker\!\left(\mathcal O_{(A,F)}\right).
		$
	\end{enumerate}
\end{lemma}

\begin{proof}
	The proof follows by duality from Lemma~\ref{lemma2}.
\end{proof}	

The following lemma shows that Functional Observability is equivalent to the
ability to uniquely determine the initial functional value from the measured
input-output data over a finite time interval.

\begin{lemma}[Behavioural Characterization of FO]\label{lem:FO-behavioural}
	The following statements are equivalent:
	\begin{enumerate}
		\item[(i)] There exists a finite time \(T>0\) such that the initial
		functional value \(z(0)=Fx(0)\) can be uniquely determined from the
		output history \(y(t)\), the known input \(u(t)\), and
		\(0\leq t\leq T\) (originally in~\cite{ref3n}).
		
		\item[(ii)]
		$
		\ker\!\left(\mathcal O_{(A,C)}\right)\subseteq\ker(F).
		$
	\end{enumerate}
\end{lemma}

\begin{proof}
	Assume (i), and let
	\(v\in\ker\!\left(\mathcal O_{(A,C)}\right)\).
	The initial conditions \(x_1(0)\) and \(x_1(0)+v\), under the same
	input, generate identical output histories. Hence, uniqueness of
	\(Fx(0)\) implies
	$
	Fx_1(0)=F\bigl(x_1(0)+v\bigr),
	$
	and therefore \(Fv=0\). Thus,
	$
	\ker\!\left(\mathcal O_{(A,C)}\right)\subseteq\ker(F).
	$
	
	Conversely, assume (ii). Let \(x_1(0)\) and \(x_2(0)\) be two initial
	conditions that generate identical output histories under the same known
	input. Then
$
	x_1(0)-x_2(0)\in
	\ker\!\left(\mathcal O_{(A,C)}\right).
$
	By (ii),
$
	x_1(0)-x_2(0)\in\ker(F),
$
	and therefore
$
	F\bigl(x_1(0)-x_2(0)\bigr)=\mathbf{0},
$
	or equivalently,
$
	Fx_1(0)=Fx_2(0),
$
proving (i).
\end{proof}

\begin{remark}
	Lemma~\ref{lem:FO-behavioural} shows that Functional Observability
	admits equivalent behavioural and structural formulations, unlike
	Functional Controllability, whose behavioural counterpart is the
	weaker Target Output Controllability property.
\end{remark}

The following definitions are obtained by restricting the generalized
eigenvector characterizations in items~(vi) of
Lemmas~\ref{lemma2} and~\ref{lemma3} to eigenvalues satisfying
$\Re(\lambda)\ge0$.

\begin{definition}[Functional Stabilizability] \label{def4}
	The functional \(z(t)=Fx(t)\) is stabilizable, or equivalently the triple $(A,B,F^{\mathsf T})$ is \emph{functional stabilizable}
	if and only if the generalized left eigenvector characterization in
	item~(vi) of Lemma~\ref{lemma2} holds for every generalized left
	eigenvector associated with an eigenvalue satisfying
	\(\Re(\lambda)\ge0\).
\end{definition}

\begin{definition}[Functional Detectability] \label{def5}
	The functional \(z(t)=Fx(t)\) is detectable, or equivalently the triple $(A,C,F)$ is 
	\emph{functional detectable}
	if and only if the generalized right eigenvector characterization in
	item~(vi) of Lemma~\ref{lemma3} holds for every generalized right
	eigenvector associated with an eigenvalue satisfying
	\(\Re(\lambda)\ge0\).
\end{definition}

\begin{remark}
	When \(F=I\), the functional \(z(t)=Fx(t)\) reduces to the state
	\(x(t)\). Consequently, Definitions~\ref{def1},~\ref{def3},
	~\ref{def4} and \ref{def5} reduce respectively to the standard
	geometric definitions of Controllability, 
	Observability, Stabilizability and Detectability.
\end{remark}

\section{Eigenvalue Characterizations of Functional Properties}
The following lemma establishes a key kernel identity used throughout this section. Building on this result, the subsequent theorem and corollaries provide necessary and sufficient PBH-style eigenvalue characterizations of the functional properties introduced in the previous section. Together with Lemmas~\ref{lemma2} and~\ref{lemma3}, they also provide equivalent characterizations based on generalized eigenvectors and generalized eigenspaces.

\begin{lemma}
	\label{lem:Plambda-kernel}
	For each $\lambda\in\sigma(A)$, let $\nu_\lambda$ denote the size of the
	largest Jordan block of $A^{\mathsf T}$ associated with $\lambda$, define
	\[
	\mathcal X_\lambda
	:=
	\ker\!\left((\lambda I-A^{\mathsf T})^{\nu_\lambda}\right),
	\]
	and 
	\[
	\mathcal P_\lambda
	:=
	\begin{pmatrix}
		(\lambda I-A^{\mathsf T})^{\nu_\lambda}\\[2pt]
		B^{\mathsf T}(\lambda I-A^{\mathsf T})^{0}\\
		\vdots\\
		B^{\mathsf T}(\lambda I-A^{\mathsf T})^{\nu_\lambda-1}
	\end{pmatrix}.
	\]
	Then
	\[
	\ker(\mathcal P_\lambda)
	=
	\mathcal X_\lambda
	\cap
	\ker\!\left(\mathcal C_{(A,B)}^{\mathsf T}\right)
	=
	\ker
	\begin{pmatrix}
		(\lambda I-A^{\mathsf T})^{\nu_\lambda}\\
		\mathcal C_{(A,B)}^{\mathsf T}
	\end{pmatrix}.
	\]
\end{lemma}

\begin{proof}
	Let
	\(
	N_\lambda:=\lambda I-A^{\mathsf T}.
	\)
	By the binomial theorem,
	\(
	N_\lambda^k
	=
	\sum_{i=0}^{k}
	(-1)^i
	\binom{k}{i}
	\lambda^{k-i}(A^{\mathsf T})^i,
	\)
	so
	\[
	\ker\!\left(\mathcal C_{(A,B)}^{\mathsf T}\right)
	=
	\ker
	\begin{pmatrix}
		B^{\mathsf T}\\
		B^{\mathsf T}N_\lambda\\
		\vdots\\
		B^{\mathsf T}N_\lambda^{\,n-1}
	\end{pmatrix}.
	\]
Since \(N_\lambda^{\nu_\lambda}x=\mathbf{0}\) for every
\(x\in\mathcal X_\lambda\), it follows that
\(N_\lambda^k x=\mathbf{0}\) for all \(k\ge\nu_\lambda\).
Hence, on \(\mathcal X_\lambda\), the rows
\(B^{\mathsf T}N_\lambda^k\), \(k\ge\nu_\lambda\),
vanish, and therefore
\[
\ker(\mathcal P_\lambda)
=
\mathcal X_\lambda
\cap
\ker\!\left(\mathcal C_{(A,B)}^{\mathsf T}\right).
\]
The second equality follows from
\[
\ker
\begin{pmatrix}
	N_\lambda^{\nu_\lambda}\\
	\mathcal C_{(A,B)}^{\mathsf T}
\end{pmatrix}
=
\ker(N_\lambda^{\nu_\lambda})
\cap
\ker(\mathcal C_{(A,B)}^{\mathsf T}),
\]
together with the definition
$
\mathcal X_\lambda
=
\ker(N_\lambda^{\nu_\lambda}).
$
\end{proof}

\begin{thm}[Eigenvalue Characterization of FC]
	\label{thm:FC-main}
	The triple $(A,B,F^{\mathsf T})$ is functional controllable if and only if, for every
	$\lambda\in\sigma(A)$,
	\begin{equation}
		\rank
		\begin{pmatrix}
			\mathcal P_\lambda\\
			F
		\end{pmatrix}
		=
		\rank(\mathcal P_\lambda). \nonumber 
	\end{equation}
\end{thm}

\begin{proof}
	Define
	$$
	\mathcal U_{\mathrm{uc},\lambda} := \ker(\mathcal P_\lambda)
	$$
	i.e., the uncontrollable component of the generalized eigenspace associated with $\lambda$.
	By Lemma~\ref{lem:Plambda-kernel},
	\[
	\ker(\mathcal P_\lambda)
	=
	\mathcal X_\lambda
	\cap
	\ker\!\left(\mathcal C_{(A,B)}^{\mathsf T}\right)
	=
	\mathcal U_{\mathrm{uc},\lambda}.
	\]
	Since the uncontrollable subspace is
	\(A^{\mathsf T}\)-invariant, it admits the primary decomposition
	\[
	\ker\!\left(\mathcal C_{(A,B)}^{\mathsf T}\right)
	=
	\bigoplus_{\lambda\in\sigma(A)}
	\mathcal U_{\mathrm{uc},\lambda}.
	\]
	Therefore, the Functional Controllability condition
	\[
	\ker\!\left(\mathcal C_{(A,B)}^{\mathsf T}\right)
	\subseteq
	\ker(F)
	\]
	holds if and only if
	\[
	\ker(\mathcal P_\lambda)
	\subseteq
	\ker(F),
	\quad
	\forall\,\lambda\in\sigma(A).
	\]
	Thus, Functional Controllability is equivalent to requiring that
	\(F\) annihilates the uncontrollable component of every generalized
	eigenspace. 
	Finally, by Lemma~\ref{lem:im-ker-duality},
	$
	\ker(\mathcal P_\lambda)
	\subseteq
	\ker(F)
	$
	is equivalent to
	\[
	\rank
	\begin{pmatrix}
		\mathcal P_\lambda\\
		F
	\end{pmatrix}
	=
	\rank(\mathcal P_\lambda),
	\]
	for every \(\lambda\in\sigma(A)\), proving the result.
\end{proof}

\begin{remark}
	Theorem~\ref{thm:FC-main} shows that Functional Controllability is
	equivalent to requiring $F$ to annihilate the uncontrollable component
	of every generalized eigenspace. Equivalently,
	\[
	\rank
	\begin{pmatrix}
		(\lambda I-A^{\mathsf T})^{\nu_\lambda}\\
		\mathcal C_{(A,B)}^{\mathsf T}\\
		F
	\end{pmatrix}
	=
	\rank
	\begin{pmatrix}
		(\lambda I-A^{\mathsf T})^{\nu_\lambda}\\
		\mathcal C_{(A,B)}^{\mathsf T}
	\end{pmatrix},
	\,\,
	\forall\,\lambda\in\sigma(A).
	\]
\end{remark}

\begin{corol}[Eigenvalue Characterization of FS]
	\label{cor:FS-main}
	Let \(z(t)=Fx(t)\), and let \(\mathcal P_\lambda\) be defined as in
	Theorem~\ref{thm:FC-main}. Then \(z(t)=Fx(t)\) is functional
	stabilizable if and only if, for every
	\(\lambda\in\sigma(A)\) satisfying \(\Re(\lambda)\geq0\),
	\begin{equation}
		\rank
		\begin{pmatrix}
			\mathcal P_\lambda\\
			F
		\end{pmatrix}
		=
		\rank(\mathcal P_\lambda). \nonumber 
	\end{equation}
\end{corol}

\begin{proof}
	The result follows directly from Definition~\ref{def4} and
	Theorem~\ref{thm:FC-main} by restricting the eigenvalue-wise
	condition to \(\lambda\in\sigma(A)\) satisfying
	\(\Re(\lambda)\geq0\).
\end{proof}

\begin{corol}[Eigenvalue Characterization of FO]
	\label{cor:FO-main}
	Let \(z(t)=Fx(t)\). For each \(\lambda\in\sigma(A)\), let
	\(\nu_\lambda\) denote the size of the largest Jordan block of
	\(A\) associated with \(\lambda\), and define
	\[
	\mathcal Q_\lambda
	:=
	\begin{pmatrix}
		(\lambda I-A)^{\nu_\lambda}\\[2pt]
		C(\lambda I-A)^0\\
		\vdots\\
		C(\lambda I-A)^{\nu_\lambda-1}
	\end{pmatrix}.
	\]
	Then the functional \(z(t)=Fx(t)\) is observable, or equivalently, the triple $(A,C,F)$ is \emph{functional observable} 
	if and only if, for every \(\lambda\in\sigma(A)\),
	\begin{equation}
		\rank
		\begin{pmatrix}
			\mathcal Q_\lambda\\
			F
		\end{pmatrix}
		=
		\rank(\mathcal Q_\lambda). \nonumber 
	\end{equation}
\end{corol}

\begin{proof}
	The result follows directly from Theorem~\ref{thm:FC-main} applied to
	the dual pair \((A^{\mathsf T},C^{\mathsf T})\). Under this duality,
	the matrix \(\mathcal P_\lambda\) becomes \(\mathcal Q_\lambda\),
	yielding the stated rank condition.
\end{proof}

\begin{remark}
	Corollary~\ref{cor:FO-main} is the dual counterpart of
	Theorem~\ref{thm:FC-main}. The PBH-style characterisation of Corollary~\ref{cor:FO-main} generalizes the result of~\cite{ref8n} by providing a necessary and sufficient characterization for arbitrary finite-dimensional linear systems, without requiring the additional structural assumptions imposed in~\cite{ref8n}, such as assumptions on the Jordan structure or additional conditions on a transformed functional matrix.
\end{remark}

\begin{remark}
	By Lemma~\ref{lem:Plambda-kernel} applied to the dual pair
	\((A^{\mathsf T},C^{\mathsf T})\),
	\[
	\ker(\mathcal Q_\lambda)
	=
	\ker
	\begin{pmatrix}
		(\lambda I-A)^{\nu_\lambda}\\
		\mathcal O_{(A,C)}
	\end{pmatrix}.
	\]
	Equivalently, the triple $(A,C,F)$ is \emph{functional observable} if and only if
	\[
	\rank
	\begin{pmatrix}
		(\lambda I-A)^{\nu_\lambda}\\
		\mathcal O_{(A,C)}\\
		F
	\end{pmatrix}
	=
	\rank
	\begin{pmatrix}
		(\lambda I-A)^{\nu_\lambda}\\
		\mathcal O_{(A,C)}
	\end{pmatrix},
	\,\,
	\forall\,\lambda\in\sigma(A).
	\]
\end{remark}

\begin{corol}[Eigenvalue Characterization of FD]
	\label{cor:FD-main}
	The triple $(A,C,F)$ is \emph{functional detectable}  if and only if, for every
	\(\lambda\in\sigma(A)\) satisfying \(\Re(\lambda)\geq0\),
	\begin{equation}
		\rank
		\begin{pmatrix}
			\mathcal Q_\lambda\\
			F
		\end{pmatrix}
		=
		\rank(\mathcal Q_\lambda). \nonumber 
	\end{equation}
\end{corol}

\begin{proof}
	The result follows directly from Definition~\ref{def5} and
	Corollary~\ref{cor:FO-main} by restricting the eigenvalue-wise
	condition to \(\lambda\in\sigma(A)\) satisfying
	\(\Re(\lambda)\geq0\).
\end{proof}

\begin{remark}
	When $F=I$, Theorem~\ref{thm:FC-main} and Corollaries~\ref{cor:FS-main},
	\ref{cor:FO-main}, and~\ref{cor:FD-main} reduce to generalized-eigenspace
	formulations of the classical PBH tests. Indeed, since
	$\ker(I)=\{{\bf0}\}$, the proposed conditions require
	$
	\mathcal U_{\mathrm{uc},\lambda}=\{{\bf0}\},$ and
	$\mathcal U_{\mathrm{uo},\lambda}=\{{\bf0}\},
	$
	for every relevant eigenvalue $\lambda$, where
	$
	\mathcal U_{\mathrm{uc},\lambda}
	=
	\mathcal X_\lambda
	\cap
	\ker\!\left(\mathcal C_{(A,B)}^{\mathsf T}\right),
	$
	and
	$
	\mathcal U_{\mathrm{uo},\lambda}
	=
	\ker \left((\lambda I-A)^{\nu_\lambda}\right)
	\cap
	\ker\!\left(\mathcal O_{(A,C)}\right).
	$
	These conditions are equivalent to the classical PBH rank tests, since a
	generalized eigenspace contains an uncontrollable (respectively,
	unobservable) generalized eigenvector if and only if it contains an
	uncontrollable (respectively, unobservable) eigenvector.
\end{remark}

The classical PBH theorem determines whether each eigenspace is completely controllable or observable. The present paper generalizes this principle by characterizing, for each generalized eigenspace, precisely which uncontrollable or unobservable components must be annihilated by a prescribed functional.

\section{Eigenstructure Characterization of Target Output Controllability}
The eigenvalue characterizations developed in the previous section rely
on the decomposition of the uncontrollable subspace into generalized
eigenspaces. In contrast, Target Output Controllability does not, in
general, admit an independent eigenvalue-wise characterization because
it depends on the intersection of
$\im(F^{\mathsf T})$ with the entire uncontrollable
subspace. The following development therefore establishes an equivalent
eigenstructure characterization based on the complete uncontrollable
subspace.

For each $\lambda\in\sigma(A)$, let
$V_{\mathrm{uc},\lambda}$ be a full-column-rank matrix whose columns
form a basis of $\ker(\mathcal P_\lambda)$. Let
$\lambda_1,\ldots,\lambda_s$ denote the distinct eigenvalues of $A$,
and define
\[
V_{\mathrm{uc}}
:=
\begin{pmatrix}
	V_{\mathrm{uc},\lambda_1}
	&
	\cdots
	&
	V_{\mathrm{uc},\lambda_s}
\end{pmatrix}.
\]
By Theorem~\ref{thm:FC-main},
\[
\ker(\mathcal P_\lambda)
=
\ker\!\left(\mathcal C_{(A,B)}^{\mathsf T}\right)
\cap
\ker\!\left(
(\lambda I-A^{\mathsf T})^{\nu_\lambda}
\right),
\]
and hence, by the primary decomposition of the uncontrollable
subspace,
\[
\im(V_{\mathrm{uc}})
=
\ker\!\left(\mathcal C_{(A,B)}^{\mathsf T}\right).
\]
Let
\[
d_{\mathrm{uc}}
:=
\rank(V_{\mathrm{uc}})
=
\dim\ker\!\left(\mathcal C_{(A,B)}^{\mathsf T}\right).
\]

\begin{thm}[Eigenstructure Characterization of TOC]
	\label{thm:TOC-main}
	Assume that $F\in\mathbb R^{r\times n}$ has full row rank.
	The following statements are equivalent.
	
	\begin{enumerate}
		\item[(i)] The triple $(A,B,F)$ is target output controllable.
		
		\item[(ii)]
		$
		\rank(F\mathcal C_{(A,B)})=\rank(F)=r.
		$
		
		\item[(iii)]
		$
		\ker\!\left(
		\mathcal C_{(A,B)}^{\mathsf T}F^{\mathsf T}
		\right)
		=
		\{\mathbf{0}\}.
		$
		
		\item[(iv)] There exists no nonzero vector
		$\alpha\in\mathbb R^r$ such that
		$
		F^{\mathsf T}\alpha
		\in
		\ker\!\left(
		\mathcal C_{(A,B)}^{\mathsf T}
		\right).
		$
		
		\item[(v)]
	$
		\rank
		\begin{pmatrix}
			F^{\mathsf T} & V_{\mathrm{uc}}
		\end{pmatrix}
		=
		r+d_{\mathrm{uc}}.
	$
	\end{enumerate}
\end{thm}

\begin{proof}
	The equivalence of (i) and (ii) is the standard rank
	characterization of target output controllability;
	see~\cite{ref5n}. Since $F$ has full row rank,
	$\rank(F)=r$.
	
	Condition (ii) holds if and only if
	$F\mathcal C_{(A,B)}$ has full row rank, or equivalently,
	\[
	\ker\!\left(
	(F\mathcal C_{(A,B)})^{\mathsf T}
	\right)
	=
	\{\mathbf{0}\}.
	\]
	Since
	\[
	(F\mathcal C_{(A,B)})^{\mathsf T}
	=
	\mathcal C_{(A,B)}^{\mathsf T}F^{\mathsf T},
	\]
	conditions (ii) and (iii) are equivalent.
	
	Moreover,
	\[
	\ker\!\left(
	\mathcal C_{(A,B)}^{\mathsf T}F^{\mathsf T}
	\right)
	=
	\left\{
	\alpha\in\mathbb R^r:
	F^{\mathsf T}\alpha
	\in
	\ker\!\left(
	\mathcal C_{(A,B)}^{\mathsf T}
	\right)
	\right\}.
	\]
	Hence, (iii) holds if and only if there exists no nonzero
	$\alpha\in\mathbb R^r$ such that
$
	F^{\mathsf T}\alpha
	\in
	\ker\!\left(
	\mathcal C_{(A,B)}^{\mathsf T}
	\right),
$
	which proves the equivalence of (iii) and (iv).
	
	Since $F$ has full row rank, $F^{\mathsf T}$ has full column rank.
	Thus, $F^{\mathsf T}\alpha\neq\mathbf{0}$ for every
	$\alpha\neq\mathbf{0}$, and condition (iv) is equivalent to
	\[
	\im(F^{\mathsf T})
	\cap
	\ker\!\left(
	\mathcal C_{(A,B)}^{\mathsf T}
	\right)
	=
	\{\mathbf{0}\}.
	\]
	By the construction of $V_{\mathrm{uc}}$,
	\[
	\im(V_{\mathrm{uc}})
	=
	\ker\!\left(
	\mathcal C_{(A,B)}^{\mathsf T}
	\right).
	\]
	Therefore, condition (iv) is equivalent to
	\[
	\im(F^{\mathsf T})
	\cap
	\im(V_{\mathrm{uc}})
	=
	\{\mathbf{0}\}.
	\]
	Since
$
	\rank(F^{\mathsf T})=r,$ and 
	$
	\rank(V_{\mathrm{uc}})=d_{\mathrm{uc}},
$
	the two subspaces have trivial intersection if and only if
	\[
	\rank
	\begin{pmatrix}
		F^{\mathsf T} & V_{\mathrm{uc}}
	\end{pmatrix}
	=
	r+d_{\mathrm{uc}}.
	\]
	This establishes (iv)$\Leftrightarrow$(v), completing the proof.
\end{proof}

\begin{remark}
	The general characterization necessarily involves the complete uncontrollable
	subspace represented by $V_{\mathrm{uc}}$, rather than its intersection with
	individual eigenspaces of $A^{\mathsf T}$ as in the case-conditioned eigenspace
	test of \cite{ref8n}. In particular, the characterization above requires neither
	a case distinction on the controllability of $(A,B)$ nor any additional
	structural condition beyond the rank test itself.
\end{remark}

\section{Intrinsic Functional Properties}

Intrinsic Functional Controllability characterizes functionals whose
dynamics can be embedded in a controllable intrinsic functional
realization. An \emph{intrinsic functional realization} is a functional
augmentation whose dynamics depend only on the augmented functional and
the input, and are therefore independent of the remaining state
components. Such realizations form the basis for functional controller
synthesis. We first define an intrinsic functional realization and then
characterize Intrinsic Functional Controllability in terms of the
existence of an admissible augmentation.

\begin{definition}
	An \emph{intrinsic functional realization} of \(z(t)=Fx(t)\) is a functional
	augmentation \(\bar z(t)=\bar F x(t)\), with \(\bar F\) formed by appending
	rows to \(F\) so that \(z(t)\) is a sub-vector of \(\bar z(t)\), whose
	dynamics are closed under \(\bar z(t)\) and the input \(u(t)\), independent
	of the remaining state components; that is,
	\[
	\dot{\bar z}=A_{\bar z}\bar z+B_{\bar z}u,
	\]
	for some matrices \(A_{\bar z}\) and \(B_{\bar z}\).
\end{definition}

\begin{definition}[Intrinsic Functional Controllability]
	\label{def:IFC}
	The functional \(z(t)\) is intrinsically controllable, or
	equivalently the triple \((A,B,F^{\mathsf T})\) is \emph{intrinsically
		functional controllable}, if there exists an intrinsic functional
	realization \(\bar z(t)\), with \(z(t)\) a sub-vector of \(\bar z(t)\),
	that can be driven to the origin with an arbitrarily assignable rate of
	convergence.
\end{definition}

The following theorem characterizes Intrinsic Functional
Controllability in terms of the existence of an admissible augmentation
\(\bar F=(\,F^{\mathsf T}\;R_1^{\mathsf T}\,)^{\mathsf T}\).

\begin{thm}[Augmentation Characterization of IFC]
	\label{thm:IFC-augmentation}
	The functional \(z(t)=Fx(t)\) is intrinsically functional
	controllable if and only if there exist an integer
	\(d\in\{r,\ldots,n\}\) and a matrix
	\(R_1\in\mathbb R^{(d-r)\times n}\)
	(with \(R_1=\emptyset\) if \(d=r\)) such that
	\[
	\bar F:=
	\begin{pmatrix}
		F\\
		R_1
	\end{pmatrix},
	\qquad
	\rank(\bar F)=d,
	\]
	and
	\begin{IEEEeqnarray}{rcl}
		\rank
		\begin{pmatrix}
			\bar FA\\
			\bar F
		\end{pmatrix}
		&=&
		\rank(\bar F),
		\nonumber\\
		\rank
		\begin{pmatrix}
			\lambda\bar F-\bar FA & \bar FB
		\end{pmatrix}
		&=&
		\rank(\bar F),
		\qquad
		\forall\lambda\in\mathbb C.
		\nonumber
	\end{IEEEeqnarray}
\end{thm}

\begin{proof}
	Since \(\bar F\) has full row rank,
$
	\bar F\bar F^{-}=I.
$
	Hence,
	\[
	\bar F\dot x
	=
	\bar FA\bar F^{-}\bar Fx
	+
	\bar FA(I-\bar F^{-}\bar F)x
	+
	\bar FBu.
	\]
	
	The augmentation \(\bar z=\bar Fx\) is an intrinsic functional realization if and only if
$
	\bar FA(I-\bar F^{-}\bar F)=\mathbf{0},
$
	which is equivalent to
	\[
	\rank
	\begin{pmatrix}
		\bar FA\\
		\bar F
	\end{pmatrix}
	=
	\rank(\bar F).
	\]
	In this case,
	\[
	\dot{\bar z}
	=
	\bar FA\bar F^{-}\bar z+\bar FBu.
	\]
	
	By Definition~\ref{def:IFC}, \(z(t)\) is intrinsically 
	controllable if and only if the induced pair
	\[
	(A_{\bar z},B_{\bar z})
	=
	(\bar FA\bar F^{-},\bar FB)
	\]
	is controllable. Since
$
	\bar FA=A_{\bar z}\bar F,
$
	we have
	\[
	\lambda\bar F-\bar FA
	=
	(\lambda I-A_{\bar z})\bar F.
	\]
	Because \(\bar F\) has full row rank,
	\[
	\rank
	\begin{pmatrix}
		(\lambda I-A_{\bar z})\bar F&\bar FB
	\end{pmatrix}
	=
	\rank
	\begin{pmatrix}
		\lambda I-A_{\bar z}&\bar FB
	\end{pmatrix}.
	\]
	The PBH test therefore gives
	\[
	\rank
	\begin{pmatrix}
		\lambda\bar F-\bar FA&\bar FB
	\end{pmatrix}
	=
	\rank(\bar F),
	\qquad
	\forall\lambda\in\mathbb C,
	\]
	which completes the proof.
\end{proof}

\begin{remark}\label{rem:IFC-feedback}
	When the conditions of Theorem~\ref{thm:IFC-augmentation} hold, the pair
	$(\bar FA\bar F^{-},\bar FB)$ is controllable. Hence, there exists a feedback gain \(K_{\bar z}\) such that the control law
	\[
	u(t)=-K_{\bar z}\bar Fx(t)
	\]
	places the poles of the augmented subsystem arbitrarily. Consequently,
	$\bar z(t)$ converges to the origin with an arbitrary rate of convergence,
	and therefore
$
	z(t)\to\mathbf{0}$
as $
	t\to\infty.
$
\end{remark}

Define
\[
\mathcal V_F
:=
\im\!\left(
\mathcal C_{(A^{\mathsf T},F^{\mathsf T})}
\right)=
\sum_{k=0}^{n-1}
\im\!\left((FA^k)^{\mathsf T}\right).
\]
Then $\mathcal V_F$ is the smallest
$A^{\mathsf T}$-invariant subspace containing
$\im(F^{\mathsf T})$.

Theorem~\ref{thm:IFC-augmentation} admits the following equivalent
subspace characterization, providing a direct geometric test for
Intrinsic Functional Controllability.

\begin{thm}[Intrinsic Subspace Characterization of IFC]
	\label{thm:IFC-subspace}
	The functional \(z(t)=Fx(t)\) is intrinsically
	controllable, or equivalently the triple $(A,B,F^{\mathsf T})$ is
	\emph{intrinsically functional controllable}, if and only if
	\[
	\mathcal V_F
	\cap
	\ker\!\left(\mathcal C_{(A,B)}^{\mathsf T}\right)
	=
	\{\mathbf{0}\}.
	\]
	Equivalently,
	$
	\rank\!\left(
	\mathcal C_{(A,B)}^{\mathsf T}
	\mathcal C_{(A^{\mathsf T},F^{\mathsf T})}
	\right)
	=
	\rank\!\left(
	\mathcal C_{(A^{\mathsf T},F^{\mathsf T})}
	\right).
	$
\end{thm}

\begin{proof}
	Suppose the triple $(A,B,F^{\mathsf T})$ is intrinsically functional controllable. Then there
	exists a controllable intrinsic functional realization
	\(\bar z=\bar Fx\) with \(z(t)\) a sub-vector of \(\bar z(t)\). Let
	\[
	\mathcal W:=\im(\bar F^{\mathsf T}).
	\]
	Since \(\bar FA=A_{\bar z}\bar F\), \(\mathcal W\) is
	\(A^{\mathsf T}\)-invariant and contains \(\im(F^{\mathsf T})\).
	Therefore,
$
	\mathcal V_F\subseteq\mathcal W.
$
	Moreover,
	\[
	\mathcal C_{(A_{\bar z},\bar FB)}
	=
	\bar F\mathcal C_{(A,B)}.
	\]
	Controllability of \((A_{\bar z},\bar FB)\) is therefore equivalent to
	\[
	\mathcal W
	\cap
	\ker\!\left(
	\mathcal C_{(A,B)}^{\mathsf T}
	\right)
	=
	\{\mathbf{0}\}.
	\]
	Since \(\mathcal V_F\subseteq\mathcal W\), it follows that
	$
	\mathcal V_F
	\cap
	\ker\!\left(
	\mathcal C_{(A,B)}^{\mathsf T}
	\right)
	=
	\{\mathbf{0}\}.
	$
	
	Conversely, suppose
$
	\mathcal V_F
	\cap
	\ker\!\left(
	\mathcal C_{(A,B)}^{\mathsf T}
	\right)
	=
	\{\mathbf{0}\}.
	$
	Choose a full-row-rank matrix \(\bar F\) satisfying
	\[
	\im(\bar F^{\mathsf T})=\mathcal V_F.
	\]
	Since \(\mathcal V_F\) is \(A^{\mathsf T}\)-invariant, there exists
	\(A_{\bar z}\) such that
	\[
	\bar FA=A_{\bar z}\bar F.
	\]
	Hence, \(\bar z=\bar Fx\) is an intrinsic functional realization
	containing \(z(t)\). Furthermore,
	\[
	\im(\bar F^{\mathsf T})
	\cap
	\ker\!\left(
	\mathcal C_{(A,B)}^{\mathsf T}
	\right)
	=
	\{\mathbf{0}\}.
	\]
	Since \(\bar F^{\mathsf T}\) has full column rank, this is equivalent to
	\[
	\rank\!\left(
	\bar F\mathcal C_{(A,B)}
	\right)
	=
	\rank(\bar F).
	\]
	Thus, \((A_{\bar z},\bar FB)\) is controllable, and \(z(t)\) is
	intrinsically controllable.
	
	Finally, the equivalent rank condition follows from
$
	\mathcal V_F
	=
	\im\!\left(
	\mathcal C_{(A^{\mathsf T},F^{\mathsf T})}
	\right).
$
\end{proof}

Intrinsic Functional Stabilizability is the natural counterpart of
Intrinsic Functional Controllability: the intrinsic realization must
still have closed functional dynamics, but stabilizability is required
only for modes with eigenvalues in the closed right-half plane.
Let
$
\mathcal X_{+}
$
denote the generalized eigenspace of \(A^{\mathsf T}\) associated with
all eigenvalues satisfying \(\Re(\lambda)\ge0\).

\begin{definition}[Intrinsic Functional Stabilizability]
	\label{def:IFS}
	The functional \(z(t)\) is intrinsically stabilizable, or equivalently
	the triple \((A,B,F^{\mathsf T})\) is \emph{intrinsically functional
		stabilizable}, if there exists an intrinsic functional realization
	\(\bar z(t)\), with \(z(t)\) a sub-vector of \(\bar z(t)\), that can be
	asymptotically driven to the origin.
\end{definition}

\begin{corol}[Intrinsic Subspace Characterization of IFS]
	\label{cor:IFS-subspace}
	The triple \((A,B,F^{\mathsf T})\) is \emph{intrinsically functional
		stabilizable} if and only if
	$
	\mathcal V_F
	\cap\mathcal X_{+}
	\cap
	\ker\!\left(\mathcal C_{(A,B)}^{\mathsf T}\right)
	=
	\{\mathbf{0}\}.
	$
\end{corol}

\begin{proof}
	Choose a full-row-rank matrix \(\bar F\) such that
	$
	\im(\bar F^{\mathsf T})=\mathcal V_F.
	$
	Since \(\mathcal V_F\) is \(A^{\mathsf T}\)-invariant, there exists
	\(A_{\bar z}\) such that
	$
	\bar FA=A_{\bar z}\bar F,
$
	and hence
	$
	A^{\mathsf T}\bar F^{\mathsf T}
	=
	\bar F^{\mathsf T}A_{\bar z}^{\mathsf T}.
	$
	Therefore, the nonstable modes of the induced pair
	\((A_{\bar z},\bar FB)\) correspond to those contained in
	\(\mathcal V_F\cap\mathcal X_{+}\). Moreover, its uncontrollable
	subspace corresponds to
	$
	\mathcal V_F
	\cap
	\ker\!\left(\mathcal C_{(A,B)}^{\mathsf T}\right).
	$
	The induced pair is stabilizable if and only if its uncontrollable
	subspace contains no nonzero generalized eigenvector associated with
	an eigenvalue satisfying \(\Re(\lambda)\ge0\). This is equivalent to
	$
	\mathcal V_F
	\cap\mathcal X_{+}
	\cap
	\ker\!\left(\mathcal C_{(A,B)}^{\mathsf T}\right)
	=
	\{\mathbf{0}\}.
	$
\end{proof}

\section{Generalized Separation Principle}

The classical separation principle states that controller and observer
design decouple for controllable and observable systems.
The following theorem establishes an analogous result at the functional
level.

\begin{thm}[Generalized Separation Principle]
	\label{thm:GSP}
	Suppose that
	\begin{enumerate}
		\item $(A,B,F^{\mathsf T})$ is \emph{intrinsically functional controllable}, and
		\item $(A,C,F)$ is \emph{functional observable}.
	\end{enumerate}
	Then there exist
	\begin{itemize}
		\item a functional state-feedback law
		$
		u(t)=-K_{\bar z}\widehat{\bar z}(t),
		$
		where $\bar z(t)=\bar Fx(t)$,
		and
		\item a functional observer producing an estimate $\widehat{\bar z}(t)$ of
		$\bar z(t)$, with estimation error $\bar z(t)-\widehat{\bar z}(t)\to\mathbf{0}$
		as $t\to\infty$,
	\end{itemize}
	such that
	$
	z(t)=Fx(t)\rightarrow{\bf0}$
	as 
	$t\rightarrow\infty,$	and the controller and observer can be designed independently.
\end{thm}
\begin{proof}
	By Theorem~\ref{thm:IFC-subspace}, choose a full-row-rank matrix
	\(\bar F\) containing \(F\) such that
	$
	\im(\bar F^{\mathsf T})=\mathcal V_F.
	$
	Then
	\((\bar FA\bar F^{-},\bar FB)\) is controllable. Hence, there exists a
	feedback gain \(K_{\bar z}\) such that
	\(\bar FA\bar F^{-}-\bar FBK_{\bar z}\) has arbitrarily assignable poles.
	
	Since
	$
	\im(\bar F^{\mathsf T})
	=
	\mathcal V_F
	\subseteq
	\im\!\left(
	\mathcal O_{(A,C)}^{\mathsf T}
	\right),
	$
	the triple \((A,C,\bar F)\) is functional observable. Hence, a
	functional observer exists whose error matrix \(N\) is Hurwitz.
	With the observer-based control law
	\[
	u=-K_{\bar z}\widehat{\bar z},
	\]
	the closed-loop system assumes the block upper-triangular form
	\[
	\begin{pmatrix}
		\dot{\bar z}\\
		\dot{\epsilon}
	\end{pmatrix}
	=
	\begin{pmatrix}
		\bar FA\bar F^{-}-\bar FBK_{\bar z} & *\\
		\mathbf{0} & N
	\end{pmatrix}
	\begin{pmatrix}
		\bar z\\
		\epsilon
	\end{pmatrix},
	\]
	where the derivation is analogous to that in~\cite{ref4n}. Therefore,
	the closed-loop poles are
	\[
	\sigma(\bar FA\bar F^{-}-\bar FBK_{\bar z})\cup\sigma(N),
	\]
	so the controller and observer poles can be assigned independently.
	Hence,
	$
	\bar z(t)\to\mathbf{0},
	$
	and therefore
	$
	z(t)=Fx(t)\to\mathbf{0}
	$
	as \(t\to\infty\).
\end{proof}

\begin{remark}[Asymptotic Generalized Separation Principle]
	If $(A,B,F^{\mathsf T})$ is intrinsically functional stabilizable and
	$(A,C,F)$ is functional detectable, then there exist a functional
	controller and a functional observer such that
$
	z(t)=Fx(t)\rightarrow{\bf0},
$
	where stabilization and estimation are required only for the unstable
	functional dynamics. In particular, if $(A,B,F^{\mathsf T})$ is
	intrinsically functional controllable and $(A,C,F)$ is functional
	observable, then the poles of the functional controller and functional
	observer can be assigned arbitrarily. Otherwise, only the unstable eigenvalues associated with the
	functional controllable and functional observable dynamics can be
	assigned arbitrarily, while the remaining stable eigenvalues remain
	unchanged.
\end{remark}

Theorem~\ref{thm:GSP} establishes the Generalized Separation Principle
as an intrinsic property of $(A,B,C,F)$. The augmentation matrices
$R_1$ and $R_2$ are required only for realizing the observer-based
functional controller: $R_1$ extends $F$ to the matrix $\bar F$ used in
the intrinsic functional controller, while $R_2$ further extends
$\bar F$ so that $F$, $R_1$, and $R_2$ together span the observable
subspace, as required for the functional observer of~\cite{22new}. The
following lemma provides one explicit realization.

\begin{lemma}[Explicit realization of the augmentation matrices]
	\label{lem:R1R2}
	Suppose that
	\begin{enumerate}
		\item $(A,B,F^{\mathsf T})$ is \emph{intrinsically functional controllable}, and
		\item $(A,C,F)$ is \emph{functional observable}.
	\end{enumerate}
	Define
	\[
	\bar F=
	\begin{pmatrix}
		F\\
		R_1
	\end{pmatrix},
	\]
	where the rows of $R_1$ are chosen so that
	$
	\im(\bar F^{\mathsf T})=\mathcal V_F.
	$
	Let $T$ be a nonsingular matrix yielding the observability decomposition
	\[
	T^{-1}AT=
	\begin{pmatrix}
		\bar A & \mathbf{0}\\
		A_{21} & A_u
	\end{pmatrix},
	\qquad
	CT=
	\begin{pmatrix}
		\bar C & \mathbf{0}
	\end{pmatrix},
	\]
	where $h=n-n_u$, $A_u\in\mathbb R^{n_u\times n_u}$, and
	$(\bar A,\bar C)$ is observable. Writing
$
	\bar FT=
	\begin{pmatrix}
		\bar F_o & \mathbf{0}
	\end{pmatrix},
$
	where $\bar F_o\in\mathbb R^{d\times h}$ has full row rank,
	let $\bar F_o^\perp$ be any matrix whose rows complete those of
	$\bar F_o$ to a basis of $\mathbb R^h$, and define
$
	R_2=
	\begin{pmatrix}
		\bar F_o^\perp & {\bf 0}
	\end{pmatrix}
	T^{-1}.
$
	Then $R_1$ and $R_2$ constitute an admissible realization satisfying
	the augmentation conditions of Theorem~15 in~\cite{ref4n}.
\end{lemma}

\begin{proof}
	By construction,
$
	\im(\bar F^{\mathsf T})=\mathcal V_F$,
	$\bar F=
	\begin{pmatrix}
		F\\
		R_1
	\end{pmatrix},
	$
	so conditions~(26a) and~(26b) of Theorem~15 in~\cite{ref4n}
	follow immediately from the defining properties of
	$\mathcal V_F$.
	
	Moreover,
$
	\bar FT=
	\begin{pmatrix}
		\bar F_o&\mathbf{0}
	\end{pmatrix},
$ 
$
	R_2T=
	\begin{pmatrix}
		\bar F_o^\perp&\mathbf{0}
	\end{pmatrix},
$
	and
	\(
	\begin{pmatrix}
		\bar F_o\\
		\bar F_o^\perp
	\end{pmatrix}
	\)
	is nonsingular. Hence,
	\(
	\begin{pmatrix}
		F\\
		R
	\end{pmatrix},
	\)
	where
	\(
	R=
	\begin{pmatrix}
		R_1\\
		R_2
	\end{pmatrix},
	\)
	forms a basis for the observable subspace. Since
	$(\bar A,\bar C)$ is observable, conditions~(25a) and~(25b) of
	Theorem~15 in~\cite{ref4n} are satisfied. Therefore,
	$R_1$ and $R_2$ constitute an admissible realization.
\end{proof}

\begin{remark}
	The construction of Lemma~\ref{lem:R1R2} also applies when
	$(A,B,F^{\mathsf T})$ is intrinsically functional stabilizable and
	$(A,C,F)$ is functional detectable. In this case, the observability
	decomposition can be refined as
	\[
	T^{-1}AT=
	\begin{pmatrix}
		A_o & \mathbf{0} & \mathbf{0}\\
		A_{21} & A_{us} & \mathbf{0}\\
		A_{31} & A_{32} & A_{uu}
	\end{pmatrix},
	\qquad
	CT=
	\begin{pmatrix}
		C_o & \mathbf{0} & \mathbf{0}
	\end{pmatrix},
	\]
	where $(A_o,C_o)$ is observable, the eigenvalues of $A_{us}$ lie in the
	open left-half complex plane, and the eigenvalues of $A_{uu}$ are
	unobservable and satisfy $\Re(\lambda)\ge0$. The matrix
	\[
	\bar A=
	\begin{pmatrix}
		A_o & \mathbf{0}\\
		A_{21} & A_{us}
	\end{pmatrix},
	\qquad
	\bar C=
	\begin{pmatrix}
		C_o & \mathbf{0}
	\end{pmatrix},
	\]
	is therefore detectable. Moreover, conditions~(25a), (25b), (26a), and
	(26b) of Theorem~15 in~\cite{ref4n} are required only for
	$
	\lambda\in\mathbb C, \Re(\lambda)\ge0.
	$
	Hence, the same realizations of $R_1$ and $R_2$ remain valid.
\end{remark}

\begin{remark}
	A functional observer-based functional controller can be designed
	according to the Generalized Separation Principle as follows.
	First, establish the intrinsic functional controllability of
	$(A,B,F^{\mathsf T})$ using Theorem~\ref{thm:IFC-subspace}. Then construct
	the control law
$
	u(t)=-K_{\bar z}\bar Fx(t),
$
	where
$
	\bar F=\begin{pmatrix}F\\R_1\end{pmatrix},
$
	with $R_1$ given by Lemma~\ref{lem:R1R2}, and $K_{\bar z}$ selected according to
	Remark~\ref{rem:IFC-feedback}. Next, verify the functional observability of
	$(A,C,F)$ using Corollary~\ref{cor:FO-main}. Finally, let
$
	L=
	\begin{pmatrix}
		F\\
		R_1\\
		R_2
	\end{pmatrix}
$
	be the observer matrix defined in \cite{22new}, where $R_2$ is given by
	Lemma~\ref{lem:R1R2}, and complete the observer design using the
	procedure of \cite{22new}.
\end{remark}

\section{Illustrative Example}
\textbf{Spacecraft Rendezvous:}
Consider the six-state linearized spacecraft rendezvous model near a
translunar halo orbit \cite{Jones1993},
$
\dot x=Ax+Bu,\;
y=Cx,
$
where
$
x=
(\xi,\eta,\zeta,\dot\xi,\dot\eta,\dot\zeta)^{\mathsf T},
$
and
{\small
	\[
	A=
	\begin{pmatrix}
		0&0&0&1&0&0\\
		0&0&0&0&1&0\\
		0&0&0&0&0&1\\
		7.380861&0&0&0&2&0\\
		0&-2.190431&0&-2&0&0\\
		0&0&-3.190431&0&0&0
	\end{pmatrix},
	\]
}
{\small
	\[
	B=
	\begin{pmatrix}
		B_1&B_2&B_3
	\end{pmatrix}
	=
	\begin{pmatrix}
		0&0&0\\
		0&0&0\\
		0&0&0\\
		1&0&0\\
		0&1&0\\
		0&0&1
	\end{pmatrix},
	\qquad
	C=
	\begin{pmatrix}
		I_3&0
	\end{pmatrix}.
	\]
}

The nominal system is controllable and observable.

\medskip
\noindent\textit{Actuator failure:}

Suppose the $\zeta$-axis thruster fails, giving
$
\bar B=(B_1\;\;B_2\;\;0).
$
A conventional full-state feedback controller
$u(t)=-Kx(t)$, where
{\small
	\[
	K=
	\begin{pmatrix}
		0.015&-0.253&-0.924&1.003&-0.931&-0.351\\
		0.082&-1.714&-1.257&1.462&0.045&-0.002\\
		0.867&-1.181&0.017&0.360&0.496&0.406
	\end{pmatrix},
	\]
}
stabilizing $(A,B)$ loses stability after the failure, since
$A-\bar BK$ has a real eigenvalue at
$
\lambda=0.742,
$
which lies in the open right-half plane.
The PBH test shows, following the actuator failure, the generalized eigenspaces associated with
$
\lambda=\pm j1.786178
$
become uncontrollable.
Applying Theorem~\ref{thm:FC-main}, every controllable functional matrix
must have the form
$
F=
\begin{pmatrix}
	*&*&0&*&*&0
\end{pmatrix},
$
where $*$ denotes an arbitrary real number.
Choosing
$
F=
\begin{pmatrix}
	1&1&0&1&1&0
\end{pmatrix},
$
applying Theorem~\ref{thm:IFC-subspace}, the triple $(A,B,F^{\mathsf T})$
 is intrinsically functional controllable.
An admissible augmentation is
{\small 
\[
\bar F=
\begin{pmatrix}
	F\\
	R_1
\end{pmatrix}
=
\begin{pmatrix}
	1&1&0&1&1&0\\
	0&1&0&0&0&0\\
	0&0&0&1&0&0\\
	0&0&0&0&1&0
\end{pmatrix},
\]}
where
$
\im(\bar F^{\mathsf T})=\mathcal V_F,
$
thereby realizing the smallest
$A^{\mathsf T}$-invariant subspace generated by $F$.

The functional controller
$
u(t)=-K_{\bar z}\bar Fx(t)
$
can be designed according to Remark~\ref{rem:IFC-feedback}, yielding
{\small
	\[
	K_{\bar z}=
	\begin{pmatrix}
		10.984&-13.351&-6.734&-9.951\\
		-1.118&6.324&-1.332&6.869\\
		0&0&0&0
	\end{pmatrix},
	\]
}
which assigns the poles of the intrinsic functional closed-loop
subsystem to
$
\{-1,-2,-3,-4\}.
$

\medskip
\noindent\textit{Sensor failure:}

Suppose only the $\xi$-position sensor remains available,
{\small 
\[
C=C_1=
\begin{pmatrix}
	1&0&0&0&0&0
\end{pmatrix}.
\]}

The pair $(A,C_1)$ is not observable. However,
Corollary~\ref{cor:FO-main} shows that the functional
$z(t)=Fx(t)$ remains functional observable since
$F$ annihilates the unobservable subspace
$\operatorname{span}\{e_3,e_6\}$.
Therefore, by Theorem~\ref{thm:GSP}, a functional observer and
functional controller can be designed independently so that
$
Fx(t)\rightarrow\mathbf{0},
$
even though the complete state is, in general, neither stabilizable nor
observable. Since
$
\operatorname{rank}(\bar F)
=
\operatorname{rank}\mathcal O_{(A,C_1)}=4,
$
no additional augmentation is required, i.e.,
$
R_2=\emptyset$. In
general, however, additional augmentation may be required, in which case
\(R_2\neq\emptyset\).
The functional observer is then obtained using the procedure of
\cite{22new}, with the functional matrix
$L=\bar F$ in the notation of \cite{22new}.

The resulting observer-based functional controller preserves the
assigned poles
$
\{-1,-2,-3,-4\},
$
while the uncontrollable modes remain at
$
\lambda=\pm j1.786178.
$
Hence, no closed-loop pole lies in the open right-half plane. The
functional controllable dynamics are asymptotically stabilized,
whereas the remaining out-of-plane dynamics remain bounded and
oscillatory. Consequently, although complete asymptotic stabilization of the full state is no longer achievable following the actuator failure, the spacecraft remains in the vicinity of the desired rendezvous trajectory, thereby maintaining safe operation until the actuator fault is rectified. This illustrates the practical robustness of the proposed intrinsic functional control framework.

\section{Conclusion}
This paper developed a unified PBH-style framework for functional properties by establishing necessary and sufficient eigenvalue-based conditions for Functional Controllability, Functional Stabilizability, Functional Observability, and Functional Detectability. It was shown that Target Output Controllability does not, in general, admit an independent eigenvalue-wise characterization, and a necessary and sufficient eigenstructure characterization was derived. The paper also introduced Intrinsic Functional Controllability and Intrinsic Functional Stabilizability as structural properties for augmentation-based functional controller synthesis, leading to a Generalized Separation Principle that encompasses the classical separation principle as a special case.

\end{document}